\documentclass[12pt]{amsart}
\usepackage{mathrsfs}
\usepackage{amssymb} %Extra symbols 
\usepackage{url}
\usepackage{hyperref}
\usepackage{xcolor}
\usepackage{graphicx} %Include postscript graphics
\usepackage{color}
\usepackage{hyperref}
\definecolor{refcolor}{rgb}{0.1,0.2,0.88}
\hypersetup{
    colorlinks,
    citecolor=refcolor,
    filecolor=refcolor,
    linkcolor=refcolor,
    urlcolor=refcolor
}

\theoremstyle{definition}
\newtheorem{theorem}{Theorem}[section]
\newtheorem{definition}[theorem]{Definition}
\newtheorem{proposition}[theorem]{Proposition}

\newtheorem{corollary}[theorem]{Corollary}
\newtheorem{remark}[theorem]{Remark}

\def\({\left(}
\def\){\right)}

\newcommand{\mc}[1]{\mathcal{#1}}
\newcommand{\ms}[1]{\mathscr{#1}}

\newcommand{\R}{\mathbb{R}}
\newcommand{\N}{\mathbb{N}}

\newcommand{\tn}{\textnormal}

\newcommand{\ie}{\textit{i.e.}}

\newcommand{\sref}[1]{\S\ref{#1}}

\newcommand{\lleq}{\underline{\ll}}
\newcommand{\causal}{\leq}

\newcommand{\image}[3]{
\begin{figure}[ht]
\begin{center}
\includegraphics[width=#2\textwidth]{#1}
\caption{\small{\label{#1}#3}}
\end{center}
\end{figure}
}

\makeatletter
\renewcommand\section{\@startsection {section}{1}{\z@}%
                                   {-3.5ex \@plus -1ex \@minus -.2ex}%
                                   {2.3ex \@plus.2ex}%
                                   {\center\normalfont\Large\bfseries}}
\renewcommand\subsection{\@startsection {subsection}{2}{\z@}%
                                   {-3.5ex \@plus -1ex \@minus -.2ex}%
                                   {2.3ex \@plus.2ex}%
                                   {\normalfont\large\bfseries}}
\renewcommand\subsubsection{\@startsection {subsubsection}{3}{\z@}%
                                   {-3.5ex \@plus -1ex \@minus -.2ex}%
                                   {2.3ex \@plus.2ex}%
                                   {\normalfont\bfseries}}
\makeatother

\begin{document}

%--------------------------------------------------------
% Title
\title{Spacetime causal structure and dimension from horismotic relation} 
\author{O.C. Stoica}
\date{\today. Horia Hulubei National Institute for Physics and Nuclear Engineering, Bucharest, Romania. E-mail: cristi.stoica@theory.nipne.ro, holotronix@gmail.com}

\begin{abstract}
A reflexive relation on a set can be a starting point in defining the causal structure of a spacetime in General Relativity and other relativistic theories of gravity. If we identify this relation as the relation between lightlike separated events (the horismos relation), we can construct in a natural way the entire causal structure: causal and chronological relations, causal curves, and a topology. By imposing a simple additional condition, the structure gains a definite number of dimensions. This construction works both with continuous and discrete spacetimes. The dimensionality is obtained also in the discrete case, so this approach can be suited to prove the fundamental conjecture of causal sets. Other simple conditions lead to a differentiable manifold with a conformal structure (the metric up to a scaling factor) as in Lorentzian manifolds. This structure provides a simple and general reconstruction of the spacetime in relativistic theories of gravity, which normally requires topological structure, differential structure, geometric structure (which decomposes in the conformal structure, giving the causal relations, and the volume element). Motivations for such a reconstruction come from relativistic theories of gravity, where the conformal structure is important, from the problem of singularities, and from Quantum Gravity, where various discretization methods are pursued, particularly in the causal sets approach.
\end{abstract}
%\pacs{04.20.Gz, 04.20.-q, 02.40.Ky, 04.60.Nc}

%--------------------------------------------------------
% Title and contents

\maketitle

%\setcounter{tocdepth}{2}
%\tableofcontents

%~~~~~~~~~~~~~~~~~~~~~~~~~~~~~~~~~~~~~~~~~~~~~~~~~~~~~~~~~~~~~~~~~~~~~~~%
\section{Introduction}
\label{s_intro}

In Lorentzian manifolds, the {\em causal relations} are defined as holding between events that can be joined by future oriented causal curves. Causal relations give the {\em causal structure} of a spacetime. In \cite{KronheimerPenrose1967CausalSpaces}, the causal structure was used to recover the {\em horismos} and {\em chronology relations} of a spacetime (the relations between events that can be joined by future lightlike, respectively timelike curves). The causal structure is known to be sufficient to recover the metric of the spacetime up to a conformal factor. The conformal factor can be obtained if in addition we know a measure which gives the volume element \cite{Zeeman1964CausalityLorentz,Zeeman1967TopologyMinkowski,Finkelstein1969SpaceTimeCode,Hawking1976NewTopologyST,Malament1977TopologySpaceTime}. This works for {\em distinguishing spacetimes} -- spacetimes whose events can be distinguished by the chronological relations they have with the other events (for example, spacetimes containing closed timelike curves are not distinguishing). Moreover, for distinguishing spacetimes, the causal structure can be obtained from the horismos relation \cite{Minguzzi2009HorismosGeneratesCausal}.

The fact that the causal structure and a measure are enough to recover the geometry of spacetime in General Relativity and other relativistic theories of gravity, and the hope that discretization may be the way to Quantum Gravity by providing an UV cutoff, motivated the study of sets ordered by the causal order \cite{Finkelstein1969SpaceTimeCode,tHooft1978QGCausalSets,Myrheim1978StatisticalGeometry}. Another motivation was that a discrete structure could account for the black hole entropy \cite{Bekenstein1973BHEntropy,Bardeen1973BHThermodynamics}.
These reasons led in particular to the idea of {\em causal set}, defined as a set $C$ endowed with a partial order $\causal$, which therefore is reflexive, antisymmetric and transitive (in standard causal set articles and some general relativity articles like \cite{Penrose1972TopologyInRelativity} the notation ``$\prec$'' is used, but in standard general relativity articles and textbooks like \cite{HE95} the notation ``$\causal$'' is preferred). In addition, it is required that for any $a,b\in C$, the cardinality of the set $\{p\in C|a\causal p\causal b\}$ is finite \cite{Bombelli1987SpaceTimeCausalSet,Sorkin1990SpacetimeAndCausalSets,Sorkin2005CausalSets}. In the causal set approach, the continuous spacetime is considered to be an effective limit of the causal set. The measure used to recover the volume element is given by the number of events in each region. As Sorkin put it, ``order plus number equals geometry''.

However, causal sets don't have a definite dimension. One sort of dimension is the smallest dimension of a Minkowski spacetime in which the causal set can be embedded (flat conformal dimension), but there are more possible definitions of dimension, such as statistical and spectral dimensions \cite{Meyer1988DimensionOfCausalSets,Sorkin1990SpacetimeAndCausalSets,Reid2003ManifoldDimensionCausalSet,Eichhorn2014SpectralDimCausalSet}, none of them satisfactory enough for this problem. This is probably the main reason why it is so difficult to prove and even to formulate mathematically the fundamental conjecture of causal sets ({\em Hauptvermutung}), that the causal set can recover within reasonable approximation (yet to be defined) the manifold structure. In the limiting case of infinite event density uniformly distributed, the conjecture has been proven \cite{Bombelli1989OriginLorentzianGeometry}, but in general the problem remains open. 

%Another reconstruction of General Relativity, which aims to allow the possibility of discreteness and still working with continuous, is undertaken by Tim Maudlin, who proposes to use linear structures to reconstruct geometry, and directed linear structures (curves endowed with a total order) to reconstruct the spacetime in Relativity, instead of using of the topology of open sets and differential geometry \cite{Maudlin2014LinearStructures,Maudlin15}. By this approach, a topology which captures the causality is obtained (also see \cite{Hawking1976NewTopologyST}). However, if we want to recover the general relativistic spacetime from directed linear structures, some difficulties appear. Firstly, there is a need to add axioms that fix the number of dimensions, the manifold topology, and the differential structure. The straightest way to do this seems to be to just impose local diffeomorphisms with $\R^n$, which amounts to appealing to the manifold topology. Secondly, compatibility conditions between the curves are needed, due to the redundancy of the information on the regions where the curves are overlapping.

In the following, we consider sets of events endowed with a reflexive relation which represents the horismos relation. We do this in the most general settings, including both the continuous and the discrete cases. We show that from the horismos relation one can recover the topology, the causal structure, and, with simple additional requirements, the dimension of spacetime and the differential structure.

The article is organized as follows. We start with the definitions and main properties of the horismotic sets in Section \sref{s_horismotic_sets_properties}. From them, we derive the causal structure in Section \sref{s_horismotic_sets_causal_structure}, and show how to obtain a topology in Section \sref{s_topology}. Then, based on the horismotic sets, we introduce  the causal curves in Section \sref{s_causal_curves}. To recover the dimension, we start with a simple example in two dimensions, which contains the main ingredients in Section \sref{s_dim_example}. Then, we introduce a notion of dimension on horismotic sets in Section \sref{s_dim_dim}, which allows us to construct lightcone coordinates in Section \sref{s_lightcone_coords}. These allow straightforwardly to recover the structure of a topological manifold, and under reasonable conditions, the differential structure and the conformal structure in Section \sref{s_spacetime}.

%~~~~~~~~~~~~~~~~~~~~~~~~~~~~~~~~~~~~~~~~~~~~~~~~~~~~~~~~~~~~~~~~~~~~~~~%
\section{Horismotic sets}
\label{s_horismotic_sets}

%~~~~~~~~~~~~~~~~~~~~~~~~~~~~~~~~~~~~~~~~~~~~~~~~~~~~~~~~~~~~~~~~~~~~~~~%
\subsection{Elementary properties}
\label{s_horismotic_sets_properties}

\begin{definition}(horismotic set)
\label{def_horismotic_set}
A {\em horismotic set} $(\mc M,\to)$ is a set $\mc M$ whose elements we call {\em events}, endowed with a binary relation $\to$, which is {\em reflexive} ($a \to a$ for any $a\in\mc M$). If $a\to b$, one says that $a$ and $b$ are in the {\em horismos} relation, or simply that $a$ {\em horismos} $b$.
For an event $a\in\mc M$, we define its {\em future horismos} or {\em future lightcone} as $E^+(a):=\{b\in\mc M|a \to b\}$, and its {\em past horismos} or {\em past lightcone} as $E^-(a):=\{b\in\mc M|b \to a\}$.
\end{definition}

The horismos relation $a\to b$ has the physical meaning that a light ray can be emitted from the event $a$ to $b$, thus it represents the lightlike separation between events. This relation is not transitive.

\begin{definition}
The horismotic relation $\to$ is {\em antisymmetric} ({\ie} for any two events $a$ and $b$ from $\mc M$, from $a\to b$ and $b \to a$ follows $a=b$) if and only if for any event $a$ from $\mc M$, $E^+(a)\cap E^-(a)=\{a\}$.
\end{definition}

\begin{proposition}
\label{thm_antisymmetric_distinguishing}
If the horismotic relation $\to$ is antisymmetric, then for any two events $a$ and $b$ from $\mc M$, from $E^+(a)= E^+(b)$ follows that $a=b$.
\end{proposition}
\begin{proof}
If $E^+(a)=E^+(b)$, then, since $a\to a$ and $b \to b$, it follows that $a\in E^+(a)=E^+(b)$ and $b\in E^+(b)= E^+(a)$. Hence, $a \to b$ and $b \to a$, and from antisymmetry, $a=b$.
\end{proof}

%~~~~~~~~~~~~~~~~~~~~~~~~~~~~~~~~~~~~~~~~~~~~~~~~~~~~~~~~~~~~~~~~~~~~~~~%
\subsection{Causal structure}
\label{s_horismotic_sets_causal_structure}

\begin{definition}
\label{def_relations}
A {\em horismotic chain} between two events $a,b\in \mc M$ is a set of $k+1$ events $\{c_0,\ldots,c_k\}\in\mc M$, where $k\in\N$ is a non-negative integer, so that $c_0=a$, $c_{i-1}\to c_i$ for all $i$, and $c_k=b$. The length of the chain is then defined to be $k+1$. Let's define the {\em causal relation} between two events $a,b\in\mc M$,  by $a\causal b$ iff there is a horismotic chain joining $a$ and $b$.
We define the {\em chronology relation} $\ll$ on $\mc M$ by $a \ll b$ iff $a \causal b$ and not $a\to b$. The relation $\ll$ represents timelike separation between events.
We also define the relation $\lleq$ by $a\lleq b\Leftrightarrow (a\ll b)\vee(a=b)$.
Two events $a,b\in\mc M$ are {\em spacelike separated}, $a \natural b$, iff neither $a \causal b$ nor $b \causal a$.
\end{definition}

\begin{definition}
\label{def_cones_intervals}
For an event $a\in\mc M$, we define its {\em chronological future} by $I^+(a):=\{b\in\mc M|a\ll b\}$, and its {\em chronological past} by $I^-(a):=\{b\in\mc M|b\ll a\}$. We define its {\em causal future} by $J^+(a):=\{b\in\mc M|a\causal b\}$, and its {\em causal past} by $J^-(a):=\{b\in\mc M|b\causal a\}$.
We define the {\em causal cone} of $a$ by $J(a)=J^+(a)\cup J^-(a)$, the {\em chronological cone} of $a$ by $I(a)=I^+(a)\cup I^-(a)$, and the {\em lightcone} of $a$ by $E(a)=E^+(a)\cup E^-(a)$.
We define $E_\ast(a):=E(a)\setminus\{a\}$, $E^\pm_\ast(a):=E^\pm(a)\setminus\{a\}$, $J_\ast(a):=J(a)\setminus\{a\}$, and $J^\pm_\ast(a):=J^\pm(a)\setminus\{a\}$.
Two events $a,b\in\mc M$ define a {\em chronological interval} $I(a,b):=I^+(a)\cap I^-(b)$, and a {\em causal interval} $J(a,b):=J^+(a)\cap J^-(b)$.
\end{definition}

\begin{proposition}
Let $a,b\in\mc M$. Then, $b\in J^+(a) \Leftrightarrow a\in J^-(b)$, $b\in I^+(a) \Leftrightarrow a\in I^-(b)$, and $b\in E^+(a) \Leftrightarrow a\in E^-(b)$.
\end{proposition}
\begin{proof}
Follows immediately from Definitions \ref{def_horismotic_set}, \ref{def_relations}, and \ref{def_cones_intervals}.
\end{proof}

\begin{proposition}
\label{thm_causality_transitive}
The causal relations $\causal$ and $\ll$ are {\em transitive}.
\end{proposition}
\begin{proof}
Follows immediately from the definitions of the relations $\causal$ and $\ll$.
\end{proof}

The causal relation $\causal$ is the smallest transitive extension of the horismos relation $\to$.

\begin{definition}
\label{def_distinguishing}
A horismotic set $\mc M$ is said to be {\em future (past) distinguishing at an event} $a\in\mc M$ if for any $b\in \mc M$, $b\neq a$ implies $I^+(a)\neq I^+(b)$ (respectively $I^-(a)\neq I^-(b)$). It is said to be {\em future (past) distinguishing} if it is future (past) distinguishing at all of its events. It is said to be {\em distinguishing} if it is both future and past distinguishing.
\end{definition}

Many of the properties of the causal and chronological relations known from General Relativity and Lorentzian manifolds in general \cite{Penrose1972TopologyInRelativity} can be derived in the settings of horismotic sets.

%~~~~~~~~~~~~~~~~~~~~~~~~~~~~~~~~~~~~~~~~~~~~~~~~~~~~~~~~~~~~~~~~~~~~~~~%
\subsection{The topology}
\label{s_topology}

We can endow $\mc M$ with a structure of topological space, generated by finite intersections and unions of open sets the sets of the form $I^\pm(a)$. As an example, consider the spacetime of General Relativity and other relativistic theories of gravity. The sets of the form $I^\pm(a)$ are the interiors of future and past lightcones, and are indeed open sets, and generate the {\em Alexandrov interval topology}. This topology coincides with the manifold topology iff it is Hausdorff, and iff the spacetime is {\em strongly causal} (at each event there is an open set $U$ so that timelike curves that leave $U$ don't return) \cite{Penrose1972TopologyInRelativity}.

But not any horismotic set has a definite dimension, nor it is locally homeomorphic to $\R^n$. Additional conditions are needed, and will be provided in the following.

%~~~~~~~~~~~~~~~~~~~~~~~~~~~~~~~~~~~~~~~~~~~~~~~~~~~~~~~~~~~~~~~~~~~~~~~%
\subsection{Causal curves}
\label{s_causal_curves}

To define causal curves in Lorentzian manifolds, one usually imposes conditions on the vectors tangent to the curve \cite{Penrose1972TopologyInRelativity}. However, by default a horismotic set doesn't have a differential structure, so here we will give a definition that doesn't require a differential and not even a topological structure.

\begin{definition}
\label{def_causal_curve}
Let $\triangleleft$ denote any of the relations $\to$, $\causal$, and $\lleq$ on a horismotic set $\mc M$.

An {\em open curve with respect to the relation $\triangleleft$} defined on a horismotic set $\mc M$ is a set of events $\gamma\subset\mc M$ so that the following two conditions hold
\begin{enumerate}
	\item 
the relation $\triangleleft$ is {\em total} on $\gamma$, that is, for any $a,b\in \gamma$, $a\neq b$, either $a \triangleleft b$ or $b \triangleleft a$,
	\item 
for any pair $a,b\in \gamma$, $a\triangleleft b$, if there is an event $c\in\mc M\setminus\gamma$ so that $a\triangleleft c$ and $c \triangleleft b$, then the restriction of the relation $\triangleleft$ to the set $\gamma\cup\{c\}$ is not total.
\end{enumerate}
A {\em loop or closed curve with respect to the relation $\triangleleft$} defined on a horismotic set $\mc M$ is a set of events $\lambda\subset\mc M$ so that, for any event $a\in\lambda$, the set $\lambda\setminus \{a\}$ is an open curve with respect to the same relation.

\begin{remark}
Note that usually a curve is defined as the image of a continuous injective function $\gamma:[x,y]\to M$, where $[x,y]\subset\R$, and $M$ is a topological space. Therefore, it is a topological subspace of $M$, and in the same time a totally ordered set, with the order induced by the order on the interval $[x,y]$. Definition \ref{def_causal_curve} is more general, since it applies to horismotic sets, in particular to both discrete and continuous spacetimes. In Section \sref{s_topology} the horismotic set $\mc M$ was endowed with a topology, the Alexandrov interval topology, and a curve as in Definition \ref{def_causal_curve} is still a topological subspace of $\mc M$, which has the property that it is totally ordered with respect to the relation $\triangleleft$. In the particular case when $\mc M$ is a manifold with distinguishing causal structure, the notion of curve defined here coincides with the usual notion of curve.
\end{remark}

For simplicity, in the following by ``curve'' we will understand ``open curve'', and by ``loop'', ``closed curve''.
We denote by $\ms C(\mc M,\triangleleft)$ and $\ms L(\mc M,\triangleleft)$ the set of curves, respectively loops, with respect to the relation $\triangleleft$.
Let $\gamma,\gamma'\in\ms C(\mc M,\triangleleft)$ be two curves. If $\gamma\subseteq \gamma'$, then $\gamma'$ is said to be an {\em extension} of $\gamma$, and $\gamma$ is named a {\em subcurve} of $\gamma'$. If for any extension $\gamma'$ of the curve $\gamma$ follows than $\gamma'=\gamma$, we say that $\gamma$ is an {\em inextensible curve}. If an extension $\gamma'$ of a curve $\gamma$ is inextensible, we say that $\gamma'$ is a {\em maximal extension} of $\gamma$.

If $a\in\gamma\in\ms C(\mc M,\triangleleft)$, then $a$ defines two curves for which $\gamma$ is an extension: $\gamma_{a+}:=\{b\in\gamma,a\triangleleft b\}$, and $\gamma_{a-}:=\{b\in\gamma,b\triangleleft a\}$. If $a,b\in \gamma$, $a\triangleleft b$, then they define a curve $\gamma_{ab}:=\gamma_{a+}\cap\gamma_{b-}$, and we call it the {\em segment} of the curve $\gamma$ determined by $a$ and $b$.

A curve from $\ms C(\mc M,\causal)$ is called {\em causal curve}. 
A curve from $\ms C(\mc M,\lleq)$ is called {\em chronological curve}. 
A curve from $\ms C(\mc M,\to)$ is called {\em lightlike curve}. 
Similar definitions are given for loops.
\end{definition}

\begin{remark}
It is easy to see that $\ms L(\mc M,\to)\subset\ms L(\mc M,\causal)$, $\ms L(\mc M,\lleq)\subset\ms L(\mc M,\causal)$, $\ms C(\mc M,\to)\subset\ms C(\mc M,\causal)$, and $\ms C(\mc M,\lleq)\subset\ms C(\mc M,\causal)$.
Also, if the horismos relation $\to$ is antisymmetric, then there are no closed lightlike curves, $\ms L(\mc M,\to)=\emptyset$.
Even when the relation $\to$ is antisymmetric, $\causal$ and $\lleq$ are not necessarily antisymmetric, so if we want to avoid closed causal and chronological curves, we have to add this as a condition.
\end{remark}

\begin{definition}
\label{def_causal_spacetime}
A spacetime on which there are no causal loops, that is, $\ms L(\mc M,\causal)=\emptyset$, is called {\em causal spacetime}. Similarly we define a {\em chronological spacetime} by $\ms L(\mc M,\lleq)=\emptyset$.
\end{definition}

\begin{proposition}
\label{thm_curve_in_cone}
Let $\gamma$ be a curve in $\mc M$, and $a\in\gamma$.
If $\gamma$ is a causal curve, then $\gamma_{a\pm}\subset J^\pm(a)$.
If $\gamma$ is a chronological curve, then $\gamma_{a\pm}\subset I^\pm(a)$.
If $\gamma$ is a lightlike curve, then $\gamma_{a\pm}\subset E^\pm(a)$.
\end{proposition}
\begin{proof}
Follow from the condition that the relation $\causal$ is total on any causal curve, $\lleq$ is total on any chronological curve, and $\to$ is total on any lightlike curve.
\end{proof}

\begin{corollary}
\label{thm_curve_in_intervals}
Let $\gamma$ be a curve in $\mc M$, and $a,b\in\gamma$.
If $\gamma$ is a causal curve, then $\gamma_{ab}\subset J(a,b)$.
If $\gamma$ is a chronological curve, then $\gamma_{ab}\subset I(a,b)\cup\{a,b\}$.
\end{corollary}
\begin{proof}
Follows from Proposition \ref{thm_curve_in_cone}.
\end{proof}

\begin{corollary}
\label{thm_curve_continuous}
The causal, chronological and lightlike curves and loops are continuous with respect to the interval topology defined in Section \sref{s_topology}.
\end{corollary}
\begin{proof}
From Corollary \ref{thm_curve_in_intervals}, for any causal curve $\gamma$ and $a,b\in\gamma$, the curve $\gamma_{ab}\subset J(a,b)$. Since the interiors of the intervals of the form $J(a,b)$ form a base for the interval topology, it follows that $\gamma$ is continuous. Similarly if $\gamma$ is a loop, since the curve obtained by removing a point of $\gamma$ is a continuous curve.
\end{proof}

%~~~~~~~~~~~~~~~~~~~~~~~~~~~~~~~~~~~~~~~~~~~~~~~~~~~~~~~~~~~~~~~~~~~~~~~%
\section{Recovering the dimension}
\label{s_dim}

%~~~~~~~~~~~~~~~~~~~~~~~~~~~~~~~~~~~~~~~~~~~~~~~~~~~~~~~~~~~~~~~~~~~~~~~%
\subsection{Example: the two-dimensional case}
\label{s_dim_example}

We look first at a simple example of recovering the conformal structure of a Lorentzian manifold in two-dimensions, which later will be distilled and generalized.

Assume that through any event $a\in\mc M$ pass exactly two maximal lightlike lines, say $\gamma^a_1$ and $\gamma^a_2$. 
For a Lorentzian manifold, the {\em global hyperbolicity condition} states that for any two events $a,b\in\mc M$, the set $J(a,b)$ is compact.
The notion of global hyperbolicity extends naturally to a general horismotic set $\mc M$, because it is also a topological space, as shown in Section \sref{s_topology}.
We assume that $\mc M$ satisfies global hyperbolicity.
In the two-dimensional case, this is equivalent to the condition that for any event $p\in I^+(a)$, $\gamma^a_1$ intersects $\gamma^p_2$ and $\gamma^a_2$ intersects $\gamma^p_1$ (see Fig. \ref{2dim-null.png}). Let us see why the two conditions are equivalent. If for example $\gamma^a_1$ would not intersect $\gamma^p_2$, then the set $J(a,b)$ would not be compact. Because of the assumptions at the beginning of this paragraph, the intersection $\gamma^a_1\cap\gamma^p_2$ contains a unique event $c_1$. Similarly, $\gamma^a_2\cap\gamma^p_1$ contains a unique event $c_2$.

\image{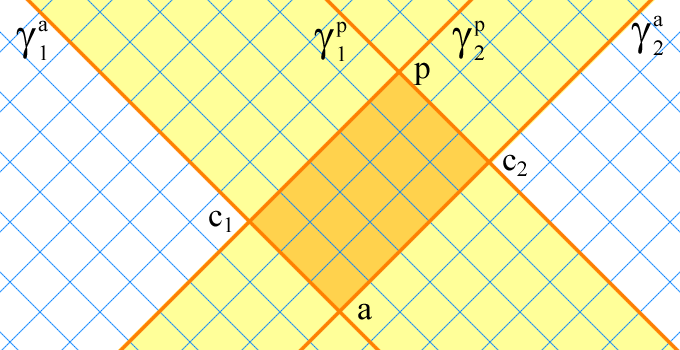}{0.5}{Any event $p\in I^+(a)$ has lightcone coordinates $(f_1(c_1),f_2(c_2))$, and $\{p\}=E^+(c_1)\cap E^+(c_2)$.}

The unique events $c_1$ and $c_2$ uniquely identify $p$. For any maximal lightlike line $\gamma$ through $a$, a parametrization $f:\gamma\to\R$ can be chosen so that for any $c,d\in \gamma$, $c\to d$, $f(c)\leq f(d)$. Then, the lightlike lines $\gamma^a_1$ and $\gamma^a_2$ together with such parametrizations will give coordinates for $I^+(a)$. In addition, if the parametrization can be chosen so that $f(\gamma)$ is an open interval in $\R$, then $\mc M$ gains a structure of topological manifold. A cover of $\mc M$ with open sets on which such lightcone coordinates are defined, and such that the transition maps are differentiable, makes $\mc M$ into a differentiable manifold of dimension $2$.

Note that there is no need to assume global hyperbolicity, a local version is enough: for any event $p\in \mc M$ there is an open set $U$ containing $p$ which is globally hyperbolic. In General Relativity and other relativistic theories of gravity, this local version is satisfied also on spacetimes that are not globally hyperbolic.

We now detail these ideas and extend them to $n$ dimensions.

%~~~~~~~~~~~~~~~~~~~~~~~~~~~~~~~~~~~~~~~~~~~~~~~~~~~~~~~~~~~~~~~~~~~~~~~%
\subsection{Dimensionality}
\label{s_dim_dim}

The mere existence of a topology defined by chronological intervals on $\mc M$ doesn't imply anything about dimension.
In order to assign to $\mc M$ a dimension, we have to define it and to require it one way or another.

\begin{definition}
\label{def_dim}
A number $n\in\N$ is called {\em dimension} of an open set $U\subset\mc M$ if there are $n$ distinct causal curves $\gamma_1,\ldots,\gamma_n\subset\mc M$
satisfying
\begin{enumerate}
	\item 
for any $p\in U$ there are $n$ events $c_i\in \gamma_i$, $i\in\{1,\ldots,n\}$ so that
\begin{equation}
\label{eq_dim}
\{p\}=\bigcap_{i=1}^n E^+(c_i)\tn{ or }\{p\}=\bigcap_{i=1}^n E^-(c_i),
\end{equation}
	\item 
	The number $n$ is the smallest with this property.
\end{enumerate}
We say that the curves $\gamma_1,\ldots,\gamma_n$ form a {\em lightcone basis of dimension $n$} of the open set $U$ (Fig. \ref{3dim-null.png}).
\end{definition}

\image{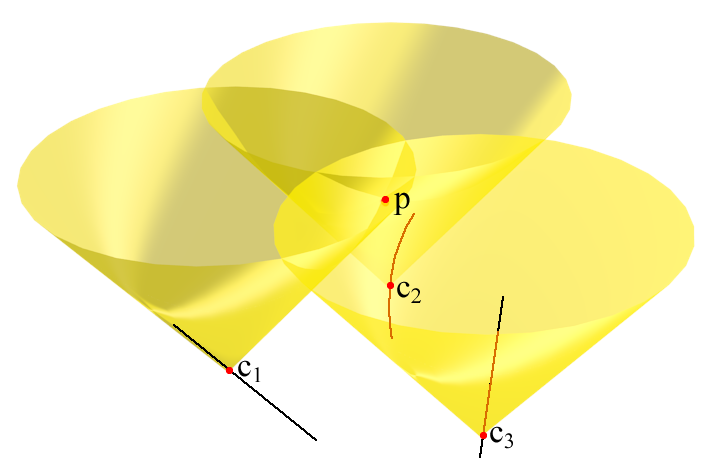}{0.5}{An event $\{p\}=E^+(c_1)\cap E^+(c_2)\cap E^+(c_3)$, with lightcone coordinates $(f_1(c_1),f_2(c_2),f_3(c_3))$.}

\begin{definition}
\label{def_dim_spacetime}
Let $n\in\N$. An {\em $n$-dimensional horismotic set} is a horismotic set $(M,\to)$ so that for any event $p\in \mc M$ there is an open set $p\in U$ of dimension $n$. We say that the dimension of $\mc M$ is $\dim \mc M=n$. 
\end{definition}

%~~~~~~~~~~~~~~~~~~~~~~~~~~~~~~~~~~~~~~~~~~~~~~~~~~~~~~~~~~~~~~~~~~~~~~~%
\subsection{Lightcone coordinates}
\label{s_lightcone_coords}

\begin{definition}
\label{def_parametrization}
A {\em parametrized causal curve} is a causal curve $\gamma$ and a function $f:\gamma\to\R$ (called {\em parametrization} of $\gamma$) which keeps the total order, that is, $f(c) \leq f(d)$ if and only if $c\causal d$. A causal curve is called {\em parametrizable} if it admits a parametrization.
\end{definition}

\begin{definition}
\label{def_lightcone_coordinates}
Let $U$ be an open set of $\mc M$, and $n$ distinct causal curves $\gamma_1,\ldots,\gamma_n\subset\mc M$, as in Definition \ref{def_dim}.
If the causal curves $\gamma_1,\ldots,\gamma_n$ are parametrized by some functions $f_1,\ldots,f_n$, the lightcone basis they form assigns to any event $p\in U$ an $n$-tuple of real numbers $(f_1(c_1),\ldots,f_n(c_n))$, hence coordinates, which we call {\em lightlike coordinates}.
\end{definition}

%~~~~~~~~~~~~~~~~~~~~~~~~~~~~~~~~~~~~~~~~~~~~~~~~~~~~~~~~~~~~~~~~~~~~~~~%
\subsection{Recovering the Lorentzian spacetime}
\label{s_spacetime}

Note that what we said so far works equally for continuous and discrete $\mc M$. Now we focus on topological manifolds.

\begin{definition}
\label{def_continuous_spacetime}
If the coordinates can be chosen to map the causal curve $\gamma$ to an open interval in $\R$, then for any event $a\in\mc M$ there is a local homeomorphism between an open set $a\in U$ and $\R^n$. In this case, $\mc M$ is called {\em continuous spacetime}.
\end{definition}

Consider an open cover $\mc U$ of $\mc M$ so that for any $U_i\in\mc U$ there is a lightcone coordinate system $f_i:U_i\to\R^n$. This endows $\mc M$ with a structure of topological manifold.

If $(\mc M,\to)$ has a structure of topological manifold, and there is an open cover $\mc U$ of $\mc M$ so that for any $U_i\in\mc U$ there is a lightcone coordinate system $f_i:U_i\to\R^n$, and so that for any $U_i\cap U_j\neq\emptyset$ the {\em transition function} $f_j|_{U_i\cap U_j}\circ(f_i|_{U_i\cap U_j})^{-1}$ is differentiable, then the cover $\mc U$ together with the charts $(U_i,f_i)$ determine a {\em differential structure} on $\mc M$.

An $n$-dimensional horismotic set $(\mc M,\to)$ whose coordinates are continuous and such that the transition maps are differentiable is naturally endowed with a structure of differentiable manifold of dimension $n$. We know the light geodesics, and they give the conformal structure of $\mc M$, that is, the metric is defined up to a scaling function $\Omega:\mc M\to(0,\infty)$.

%~~~~~~~~~~~~~~~~~~~~~~~~~~~~~~~~~~~~~~~~~~~~~~~~~~~~~~~~~~~~~~~~~~~~~~~%
\section{Discussion}
\label{s_discussion}

What is the correspondence between causal sets and horismotic sets? Given that the causal sets approach is based only on the causal relation, which does not distinguish between horismos and chronological relations, there are more ways to choose which pairs of events in causal relation are in a horismos or in a chronological relation. Moreover, if the events are selected from a continuous manifold by sprinkling, the chance that two events of the causal set are in a horismos relation is practically zero. In the continuum case, even if we start from the causal relation, the solution is unique, at least for distinguishing spacetimes: the boundary of the lightcone at an event $p$ gives the events in a horismos relation with $p$, and the interior gives those in a chronology relation with $p$. And as mentioned, the horismos relation is enough to reconstruct the causal structure \cite{Minguzzi2009HorismosGeneratesCausal}.

There are some advantages in starting with the horismos relation rather than the causal one. The chronological and causal relations can be obtained from the horismos relation. But if we start with the causal relations, as in the causal set approach, we can't obtain the horismos relation, unless for example the spacetime has a structure of topological manifold, or at least can be embedded densely in a topological manifold, which is not the case for causal sets. Even if we have the means to distinguish the horismos relation from the causal relation, we still have to impose additional compatibility constraints, since the causal relation can be generated by the horismos. %A similar redundancy is required by other approaches, like the directed linear structures approach to reconstruct the spacetime in Relativity, since compatibility conditions have to be imposed between the curves giving the causal order and topology \cite{Maudlin2014LinearStructures,Maudlin15}.
But if we start with the horismos relation, we can obtain everything about the spacetime, including the geometry up to a scaling factor, and we don't have to impose compatibility conditions, showing that the horismos relation is more fundamental.

The most important advantage of horismos relation appears to be that we can obtain a spacetime with a definite dimension by imposing a simple condition.
This may be of help in building models similar to causal sets and with definite dimension, but it may still be difficult in practice.

The properties and results that can be derived starting only with the horismos relation have correspondent in those of Lorentzian manifolds, which are presented for example in \cite{Penrose1972TopologyInRelativity}. However, we don't enter here in much detail about this, the main purpose being to recover the causal structure and the dimensionality of spacetime.

An advantage of the causal set approach is that it aims to recover in a good approximation the manifold and the conformal structure, but also to find the conformal factor needed to recover the metric, by approximating the volume with the number of events in that region. This relation between volume and number of events seems pretty clear, and at this point we don't know a way to do the same in the approach of horismotic sets, especially when the dimension is well defined as in Section \sref{s_dim_dim}. However, the volume information can be provided by a measure.

Horismotic sets in which the horismotic relation is not antisymmetric can be used to include additional structures. 
For example, consider \emph{gauge theory}, described by a fiber bundle $\mc E$ over a distinguishing Lorentzian manifold $\mc M$, and let $F$ be the typical fiber. The causal relations on $\mc M$ can be lifted to $\mc E$, in the following way. Let $p,q\in\mc E$ be  any pair of points in $\mc E$, then $p=(a,x)\in\mc M\times F$ and $q=(b,y)\in\mc M\times F$. We define the horismos relation on $\mc E$ by $p\to q$ iff $a\to b$, and similarly we define the chronological and causal relations on $\mc E$.
But then, if $a=b$ and $x\neq y$, it follows that $p\to q$ and $q\to p$ but $p\neq q$.
So even if antisymmetry holds on the base space $\mc M$, it does not hold on the bundle $\mc E$. The points of $\mc E$ in the same fiber are in the horismos (and also chronological) relation, but they are distinct. This corresponds to the gauge degrees of freedom.
Of course, the typical fiber $F$ is not simply a set, and should be endowed with additional structure, which is not captured in the horismos relation.

In short, the approach of starting with the horismos relation:
\begin{enumerate}
	\item 
is very general, because we just start with a reflexive relation, which we identify as the horismos relation;
	\item 
works for both discrete and continuous spacetimes;
	\item 
allows us to recover the causal and chronological relations, while recovering the horismotic relation from the causal relation works only in special cases, for example for continuous spacetimes;
	\item 
allows us to recover the interval topology;
	\item 
avoids redundancy and compatibility conditions when defining the causal structure, which are present when starting from the causal relation;
	\item 
allows to define causal, chronological and lightlike curves and loops, without the need of differential or even topological structures;
	\item 
allows to recover the dimensionality, as well as the manifold structure, under simple conditions (while these problems are still open in the case of causal sets).
\end{enumerate}

In another article \cite{Sto15b} it is given an additional reason to consider the causal structure as fundamental in General Relativity: while the metric becomes singular at some black hole and big bang singularities, the causal structure remains regular.

%~~~~~~~~~~~~~~~~~~~~~~~~~~~~~~~~~~~~~~~~~~~~~~~~~~~~~~~~~~~~~~~~~~~~~~~%
\subsubsection*{Acknowledgments}

The author wishes to thank the anonymous reviewers, whose comments helped improve the quality of the manuscript.

%~~~~~~~~~~~~~~~~~~~~~~~~~~~~~~~~~~~~~~~~~~~~~~~~~~~~~~~~~~~~~~~~~~~~~~~%

\end{document}